\documentclass[%
superscriptaddress,
 amsmath,amssymb,
 aps,
 pra,
floatfix,
twocolumn
]{revtex4-1}
\usepackage[utf8]{inputenc}
\usepackage[T1]{fontenc}
\usepackage{babel}
\usepackage{graphicx}
\usepackage{amsfonts}
\usepackage{amssymb}
\usepackage{amsmath}
\usepackage{amsthm}
\usepackage{mathtools}
\usepackage{physics}
\usepackage[normalem]{ulem}
\usepackage{blindtext}

\DeclareMathOperator{\arcosh}{arcosh}

\usepackage{braket}
\renewcommand\bra[1]{{\langle{#1}|}}
\renewcommand\ket[1]{{|{#1}\rangle}}
\usepackage{multirow}
\usepackage{tabularx}
\usepackage{dsfont}
\usepackage{bm}
\usepackage{xcolor}

\newtheorem{theorem}{Theorem}
\newtheorem{proposition}[theorem]{Proposition}

\theoremstyle{definition}

\usepackage{natbib}
\setcitestyle{square,numbers}
\usepackage[colorlinks]{hyperref}
\usepackage[format=plain,justification=centerlast]{caption}
\usepackage[format=plain,justification=centerlast]{subcaption}

\begin{document}

\author{Tomasz Linowski}
\affiliation{International Centre for Theory of Quantum Technologies, University of Gdansk, 80-308 Gda{\'n}sk, Poland}
\email[Corresponding author: ]{t.linowski95@gmail.com}

\author{{\L}ukasz Rudnicki}
\affiliation{International Centre for Theory of Quantum Technologies, University of Gdansk, 80-308 Gda{\'n}sk, Poland}
\affiliation{Center for Theoretical Physics, Polish Academy of Sciences, 02-668 Warszawa, Poland}

\title{Relating the Glauber-Sudarshan, Wigner and Husimi quasiprobability distributions operationally through the quantum limited amplifier and attenuator channels}

\date{\today}

\begin{abstract}
The Glauber-Sudarshan, Wigner and Husimi quasiprobability distributions are indispensable tools in quantum optics. However, although mathematical relations between them are well established, not much is known about their operational connection. In this paper, we prove that a single composition of finite-strength quantum limited amplifier and attenuator channels, known for their noise-adding properties, turns the Glauber-Sudarshan distribution of any input operator into its Wigner distribution, and its Wigner distribution into its Husimi distribution. As we dissect, the considered process, which can be performed in a quantum optical laboratory with relative ease, may be interpreted as realizing a quantum-to-classical transition.
\end{abstract}

\maketitle

\section{Introduction}
In 1932 \cite{Wigner_distribution_Wigner_1932}, Wigner discovered what is now called the Wigner function, a description of quantum mechanics akin to the classical phase space, though not fully compatible with it. A few years later, in 1940 \cite{Q_representation}, Husimi invented the Husimi representation, a function even closer to classical mechanics due to its non-negativity. Finally, in 1963 \cite{P_representation_Glauber,P_representation_Sudarshan}, Glauber and Sudarshan discovered that any quantum state of the radiation field can be written as a diagonal sum over the set of coherent states weighed according to the Glauber-Sudarshan distribution.

Today, these three quasiprobability distributions, the ``quasi'' part coming from the fact that none of them fulfill all the properties of an actual probability distribution, are some of the most important theoretical tools in quantum optics, especially in quantum tomography and investigations of non-classicality \cite{quantum_optics_summary_Olivares_2021,quasiprobability_distributions_tomography_Koczor_2020,quasiprobability_distributions_Sperling_2020,Schleich_Quantum_Optics,optics_Agarwal_2012,quasiprobability_distributions_Milburn_1989}. In the case of the former, the Wigner function has allowed for probing classes of systems as varied as squeezed vacuum \cite{Wigner_tomography_squeezed_Smithey_1993}, thermal \cite{Wigner_tomography_thermal_2021} and single photon states, among others \cite{Wigner_tomography_single_photon_Lvovsky_2001,Wigner_tomography_single_photon_Zavatta_2007,Wigner_tomography_single_photon_Laiho_2010}. In the case of the latter, both Wigner and Glauber-Sudarshan distributions are used \cite{quasiprobability_distributions_non-classicality_Tan_2020,beam_splitter_classicality_Brunelli_2015,Wigner_non-classicality_Kenfack_2004}, with the negativity of the Glauber-Sudarshan distribution typically accepted as the very definition of non-classical light \cite{Glauber_distribution_non-classicality_Titulaer_1965,Glauber_distribution_non-classicality_Mandel_1986}. The Husimi function also has many applications, ranging from the aforementioned tomography \cite{Husimi_tomography_Kanem_2005,Husimi_tomography_Agarwal_1998} to the quantum phase measurement \cite{Paul_phase_experiment_Freyberger_1993,FREYBERGER199341}.

The three phase-space distributions are not independent: the mathematical relations between them are well known \cite{optics_Agarwal_2012,quasiprobability_distributions_Sperling_2020}. More explicitly, the Wigner and Husimi distributions can be understood in terms of simple Gaussian smoothing of the Glauber-Sudarshan distribution. However, to the best of our knowledge, the physical connection between these distributions is not established, aside from the fact that the Glauber-Sudarshan distribution of a system subjected to the quantum limited amplifier channel of infinite strength approaches its Husimi distribution \cite{quantum_phase_Q_amplification_Schleich_1992}. This process, however, requires infinite energy, and is therefore unphysical.

In this paper, we propose a simple physical process that turns the Glauber-Sudarshan distribution of any operator (including any density operator) into its Wigner function, and its Wigner function into its Husimi function. This process, given by a finite-strength quantum limited amplifier followed by an attenuator channel \cite{pure-loss_De_Palma_2017,pure-loss_Holevo_2007,pure-loss_Garcia_2012,pure-loss_Mari_2014}, has a well-known operational interpretation in terms of noise addition and, being Gaussian, is readily accessible experimentally \cite{Gaussianity_resource_Albarelli_2018,Gaussianity_resource_Takagi_2018}. Furthermore, as we show, it has a number of properties expected from a ``classicalization'' procedure: for example, its double application realizes a projection onto the set of coherent states.

This paper is organized as follows. In Section \ref{sec:quasiprobabilities}, we formally introduce the discussed quasiprobability distributions. In Section \ref{sec:qla_channels} we do the same for our main tool: quantum limited amplifier and attenuator channels. In Sections \ref{sec:results} and \ref{sec:discussion}, we first derive and then thoroughly discuss our main results. Finally, we provide outlooks in Section \ref{sec:outlooks}.

\section{Glauber-Sudarshan, Wigner and Husimi quasiprobability distributions} 
\label{sec:quasiprobabilities}
We start by providing basic information about the three considered quasiprobability distributions, relevant for the derivation and understanding of our results. For a more thorough treatment, we refer the reader to one of the numerous didactic sources, e.g. \cite{quantum_optics_summary_Olivares_2021,Schleich_Quantum_Optics,optics_Agarwal_2012}. We remark that, for clarity of presentation, we assume the Hilbert space to be one-mode. However, we stress that our results apply (through a straightforward generalization) to an arbitrary number of modes.

The \emph{Glauber-Sudarshan P distribution} \cite{P_representation_Glauber,P_representation_Sudarshan} of an arbitrary operator $\hat{X}$ (including any density operator) is defined through a diagonalization in the set of coherent states $\ket{\alpha}$, i.e. the eigenstates of the annihilation operator $\hat{a}$ with eigenvalue $\alpha$. More precisely, the P distribution is such that
\begin{align} \label{eq:P}
\begin{split}
    \hat{X}
        & = \int \frac{d^2\alpha}{\pi} P_{\hat{X}}(\alpha) \ket{\alpha}\bra{\alpha},           
\end{split}
\end{align}
where
\begin{align} \label{eq:d2alpha}
\begin{split}
    \int d^2\alpha \coloneqq \int_{-\infty}^{\infty}d\Re(\alpha)\int_{-\infty}^{\infty}d\Im(\alpha).
\end{split}
\end{align}
In general, the P distribution is not a function, and in fact it can be highly singular. This is most easily seen from its formal expression in the Fock basis \cite{optics_Agarwal_2012}:
\begin{align} \label{eq:P_formal}
\begin{split}
    P_{\hat{X}}(\alpha) = \sum_{n,m=0}^\infty X_{nm} \frac{(-1)^{n+m}}{\sqrt{n!m!}} 
        e^{|\alpha|^2}\frac{\partial^{n+m}}{\partial\alpha^n\partial\alpha^{*m}}\delta(\alpha).
\end{split}
\end{align}
where $\delta(\alpha)$ is the Dirac delta distribution. For example, for a pure Fock state, we have
\begin{align} \label{eq:Fock_P}
\begin{split}
    P_{\ket{n}\bra{n}}(\alpha) = \frac{1}{n!} 
        e^{|\alpha|^2}\frac{\partial^{2n}}{\partial\alpha^n\partial\alpha^{*n}}\delta(\alpha).
\end{split}
\end{align}
Still, for states considered semi-classical, such as, e.g. thermal states, the P distribution does reduce to an ordinary, non-negative function, meaning that the state can be expressed as a classical mixture of coherent states. For this reason, the non-positivity of the P distribution is used as a criterion for non-classicality \cite{Glauber_distribution_non-classicality_Titulaer_1965,Glauber_distribution_non-classicality_Mandel_1986}. Note that, like all quasiprobability distributions, the P distribution is normalized to one.

The \emph{Wigner W distribution} \cite{Wigner_distribution_Wigner_1932} (sometimes called a function) is closer to an ordinary probability distribution: although it may obtain negative values for some states, it is never singular for them. For our purposes, it is most conveniently defined as
\begin{align} \label{eq:W}
\begin{split}
    W_{\hat{X}}(\alpha) \coloneqq \int \frac{d^2\beta}{\pi} 
        \Tr\left[\hat{X}\hat{D}(\beta)\right]e^{\alpha\beta^*-\alpha^*\beta},
\end{split}
\end{align}
where $\hat{D}(\beta)$ is the displacement operator:
\begin{align} \label{eq:displacement_operator}
\begin{split}
    \hat{D}(\beta) \coloneqq \exp\left(\beta\hat{a}^\dag-\beta^*\hat{a}\right).
\end{split}
\end{align}
The W distribution of a Fock state equals \cite{optics_Agarwal_2012}
\begin{align} \label{eq:Fock_W}
\begin{split}
    W_{\ket{n}\bra{n}}(\alpha) = 2(-1)^n e^{-2|\alpha|^2}L_n\left(4|\alpha|^2\right),
\end{split}
\end{align}
with $L_n$ being the $n$-th Laguerre polynomial. As seen, although it does take on negative values, it is still a well-behaved function. Similarly to the P distribution, non-positivity of the W function can be considered a measure of non-classicality \cite{quasiprobability_distributions_non-classicality_Tan_2020}.

Finally, the \emph{Husimi Q distribution} \cite{Q_representation} (or function) is the most classical-like: it is non-negative for all quantum states, which is obvious from its definition:
\begin{align} \label{eq:Q}
\begin{split}
    Q_{\hat{X}}(\alpha) \coloneqq \bra{\alpha}\hat{X}\ket{\alpha}.
\end{split}
\end{align}
For the particular case of a Fock state, we obtain the Poisson distribution in the particle number with mean~$|\alpha|^2$:
\begin{align} \label{eq:Fock_Q}
\begin{split}
    Q_{\ket{n}\bra{n}}(\alpha) = \frac{|\alpha|^{2n}}{n!} e^{-|\alpha|^2}.
\end{split}
\end{align}
We stress that despite its non-negativity, even the Q function is not a true probability distribution as, due to the non-orthogonality of coherent states, different values of $\alpha$ do not correspond to mutually exclusive events.

Let us remark that historically, the discussed distributions were often defined with an additional multiplicative factor of $1/\pi$. Here, we follow the more modern convention of, e.g. \cite{Wehrl_entropy_De_Palma_2017,husimi_normalization_Floerchinger_2021} and omit this factor (of course in the end both conventions give exactly the same results). In addition to resulting in more consistent formulas (the factor of $1/\pi$ now simply \emph{always} appears with the integration measure $d^2\alpha$), this arguably brings the quasiprobability distributions closer to classical distributions. Especially the Q distribution is now bounded from above by $1$, rather than $1/\pi$, meaning that, e.g. in the case of Eq. (\ref{eq:Fock_Q}) it becomes precisely the Poisson distribution, instead of only being proportional to it.

It is well known that the three discussed quasiprobability distributions are related by means of Gaussian smoothing, a so-called \emph{Weierstrass transform} \cite{optics_Agarwal_2012,quasiprobability_distributions_Sperling_2020}:
\begin{align} 
    W_{\hat{X}}(\alpha) &= \label{eq:P_into_W_transform}
        2 \int \frac{d^2\beta}{\pi} P_{\hat{X}}(\beta) e^{-2|\alpha-\beta|^2},\\
    Q_{\hat{X}}(\alpha) &= \label{eq:W_into_Q_transform}
        2 \int \frac{d^2\beta}{\pi} W_{\hat{X}}(\beta) e^{-2|\alpha-\beta|^2},
\end{align}
which, applied in succession, imply also
\begin{align} 
    Q_{\hat{X}}(\alpha) &= \label{eq:P_into_Q_transform}
        \int \frac{d^2\beta}{\pi} P_{\hat{X}}(\beta) e^{-|\alpha-\beta|^2}.
\end{align}
Eqs. (\ref{eq:P_into_W_transform}-\ref{eq:P_into_Q_transform}) provide a basic intuition about why the Q distribution behaves more classical-like than the W distribution, which in turn is more classical than the P distribution: the potential irregularities in the input distribution are being ``smeared'' in the output by a Gaussian function. The main aim of our paper is to give these mathematical formulas an explicit operational interpretation in terms of well known, experimentally available transformations: the quantum limited amplifier and the quantum limited attenuator.

\section{Quantum limited amplifier and attenuator channels}
\label{sec:qla_channels}
For a single mode, the action of the \emph{quantum limited amplifier channel} of strength $\kappa\geqslant 1$ on an operator $\hat{X}$ is defined through an interaction with an ancillary system in the vacuum state as~\cite{Quantum_limited_amplifier_Mollow_1967,QLA_De_Palma_2017}
\begin{align} \label{eq:QLAmp}
    \mathcal{A}_\kappa\big(\hat{X}\big)\coloneqq\Tr_2\left[
            \hat{S}_{12}(\kappa)
            \left(\hat{X}\otimes\ket{0}\bra{0}\right)
            \hat{S}_{12}^\dag(\kappa)\right],
\end{align}
where
\begin{align} \label{eq:squeezing_operator}
    \hat{S}_{12}(\kappa)\coloneqq
        \exp\left[\arcosh\sqrt{\kappa}(\hat{a}^\dag\hat{b}^\dag-\hat{a}\hat{b})\right]
\end{align}
is the two-mode squeezing operator and $\hat{b}$ is the annihilation operator of the ancillary system traced out in Eq. (\ref{eq:QLAmp}). The larger the value of $\kappa$, the stronger the amplification, with $\kappa=1$ corresponding to the identity channel. Physically, the amplifier describes the process of pumping particles into the system:
\begin{align} \label{eq:n_A}
    \braket{\hat{a}^\dag\hat{a}}_{\mathcal{A}_\kappa(\hat{\rho})} 
    = \kappa\braket{\hat{a}^\dag\hat{a}}_{\hat{\rho}} + \kappa-1,
\end{align}
which is easy to prove from the definition of the channel and the fact that \cite{QLA_De_Palma_2017}
\begin{align} \label{eq:squeezing_transformation_original}
    \hat{S}_{12}^\dag(\kappa) \, \hat{a} \, \hat{S}_{12}(\kappa) 
        &= \sqrt{\kappa}\,\hat{a} + \sqrt{\kappa - 1}\,\hat{b}^\dag.  
\end{align}
Importantly, the amplification occurs in a way that is associated with making the system more classical-like \cite{quantum_phase_Q_amplification_Schleich_1992}. For example, it is known that P distribution of an infinitely amplified state is always non-negative \cite{quantum_phase_Q_amplification_Schleich_1992}, i.e. semi-classical.

The action of the \emph{quantum limited attenuator channel} of strength $\lambda\in[0,1]$ on an operator $\hat{X}$ is defined similarly as~
\begin{align} \label{eq:QLAtt}
    \mathcal{E}_\lambda\big(\hat{X}\big)\coloneqq
        \Tr_2\left[\hat{B}_{12}(\lambda)\left(\hat{X}\otimes\ket{0}\bra{0}\right)
        \hat{B}_{12}^\dag(\lambda)\right],
\end{align}
where 
\begin{align}
    \hat{B}_{12}(\lambda)\coloneqq
        \exp\left[\arccos\sqrt{\lambda}(\hat{a}^\dag\hat{b}-\hat{a}\hat{b}^\dag)\right]
\end{align}
is the two-mode beamsplitter operator. Here, the strength of the channel is decreasing with $\lambda$, with the largest value $\lambda=1$ corresponding to the identity channel. Intuitively, the quantum limited attenuator, sometimes called a pure-loss channel \cite{pure-loss_Garcia_2012,pure-loss_Mari_2014}, has an inverse effect to the amplifier: it weakens the system by decreasing the number of particles within:
\begin{align} \label{eq:n_E}
    \braket{\hat{a}^\dag\hat{a}}_{\mathcal{E}_\lambda(\hat{\rho})} 
    = \lambda\braket{\hat{a}^\dag\hat{a}}_{\hat{\rho}} + 1-\lambda.
\end{align}
Again, this easily follows from the definition of the channel and the identity \cite{QLA_De_Palma_2017}
\begin{align} \label{eq:beamsplitting_transformation}
    \hat{B}_{12}^\dag(\lambda) \, \hat{a} \, \hat{B}_{12}(\lambda) 
        = \sqrt{\lambda}\,\hat{a} + \sqrt{1-\lambda}\,\hat{b}^\dag.    
\end{align}


\section{Main results}
\label{sec:results}
We are now in the position to state and prove our main result. An in-depth discussion is provided in the next section.

\begin{proposition} \label{th:main}
Let
\begin{align} \label{eq:C}
    \mathcal{C} \coloneqq \mathcal{E}_{1/2}\circ\mathcal{A}_{2},
\end{align}
    with $\mathcal{A}_\kappa$ and $\mathcal{E}_\lambda$ being the quantum limited amplifier and attenuator, respectively, as defined in the previous section. The following relations between the P, W, Q quasiprobability distributions hold:
\begin{align}
    W_{\hat{X}}(\alpha) \label{eq:P_into_W_C}
        &= P_{\mathcal{C}(\hat{X})}(\alpha),\\
    Q_{\hat{X}}(\alpha) \label{eq:W_into_Q_C}
        &= W_{\mathcal{C}(\hat{X})}(\alpha),\\ 
    Q_{\hat{X}}(\alpha) \label{eq:P_into_Q_C}
        &= P_{\mathcal{C}^2(\hat{X})}(\alpha).
\end{align}
\end{proposition}

\begin{proof}
We begin by observing that it is enough to prove only Eq. (\ref{eq:P_into_W_C}). If this equation is true, then by comparison with Eq. (\ref{eq:P_into_W_transform}) we would have
\begin{align} 
    P_{\mathcal{C}(\hat{X})}(\alpha) = 
        2 \int \frac{d^2\beta}{\pi} P_{\hat{X}}(\beta) e^{-2|\alpha-\beta|^2}.
\end{align}
Since $P_{\hat{X}}$ is completely arbitrary, we could substitute $W_{\hat{X}}$ for it, which together with Eq. (\ref{eq:W_into_Q_transform}), would yield Eq. (\ref{eq:W_into_Q_C}). The remaining Eq. (\ref{eq:P_into_Q_C}) follows directly from Eqs. (\ref{eq:P_into_W_C}, \ref{eq:W_into_Q_C}).

To prove Eq. (\ref{eq:P_into_W_C}), it will be convenient to first compute the action of the channel $\mathcal{C}$ on an arbitrary coherent state. We begin by calculating $\mathcal{A}_2(\ket{\alpha}\bra{\alpha})$. By definition (\ref{eq:QLAmp}),
\begin{align}
\begin{split}
    \mathcal{A}_2(\ket{\alpha}\bra{\alpha}) 
        &= \Tr_2 \left[\hat{S}_{12}(2) \ket{\alpha 0}\bra{\alpha 0} \hat{S}_{12}^\dag(2)\right] \\
        &= \Tr_2 \left[\hat{S}_{12}(2) \hat{D}_1(\alpha) \ket{0 0}\bra{0 0} 
            \hat{D}_1^\dag({\alpha}) \hat{S}_{12}^\dag(2)\right],
\end{split}
\end{align}
Here and below, the bottom index for the displacement operator denotes the mode it acts upon. Using the unitary property of squeezing and the fact that \footnote{To see this, observe from Eq. (\ref{eq:squeezing_operator}) that taking the hermitian conjugate of $\hat{S}_{12}(\kappa)$ is equivalent to the change $\arcosh \sqrt{\kappa} \to -\arcosh \sqrt{\kappa}$, which implies $\sqrt{\kappa-1}=\sqrt{\cosh^2(\arcosh \sqrt{\kappa})-1}=\sinh(\arcosh \sqrt{\kappa})\to -\sinh(\arcosh \sqrt{\kappa})=-\sqrt{\kappa-1}$, hence the minus sign in the r.h.s. of Eq. (\ref{eq:squeezing_action_on_annihilation_operator}) in comparison to Eq.
(\ref{eq:squeezing_transformation_original}).}
\begin{align} \label{eq:squeezing_action_on_annihilation_operator}
\begin{split}
    \hat{S}_{12}(\kappa) \, \hat{a} \, \hat{S}_{12}^\dag(\kappa)
        = \sqrt{\kappa}\,\hat{a} - \sqrt{\kappa - 1}\,\hat{b}^\dag,
\end{split}
\end{align}
we get
\begin{align} \label{eq:proof_Z_in_trace}
\begin{split}
    \mathcal{A}_2(\ket{\alpha}\bra{\alpha})
        &= \Tr_2 \left[\hat{Z}_{12}(\alpha) \hat{S}_{12}(2) \ket{00}\bra{00} 
            \hat{S}_{12}^\dag(2) \hat{Z}_{12}^\dag(\alpha)\right],
\end{split}
\end{align}
where
\begin{align} \label{eq:Z}
\begin{split}
    \hat{Z}_{12}(\alpha) &\coloneqq \hat{S}_{12}(2)\hat{D}_{1}(\alpha)\hat{S}_{12}^\dag(2)\\
    &= \exp\left[
        \alpha\left(\sqrt{2}\hat{a}^\dag - \hat{b}\right)
        -\alpha^*\left(\sqrt{2}\hat{a} - \hat{b}^\dag\right)
        \right]\\
    &= \hat{D}_{1}\left(\sqrt{2}\alpha\right)\hat{D}_{2}(\alpha^*).
\end{split}
\end{align}
Substituting this into Eq. (\ref{eq:proof_Z_in_trace}), we make two observations. Firstly, because the trace is only over the second subsystem, the displacement operators for the first subsystem can be taken outside. Secondly, because partial trace is cyclic in the subsystem that we trace out, the displacement operators in the second subsystem can be moved to cancel with each other (since they are unitary). In summary, we obtain
\begin{align}
\begin{split}
    \mathcal{A}_2&(\ket{\alpha}\bra{\alpha})\\
        &= \hat{D}_{1}\left(\sqrt{2}\alpha\right)
            \Tr_2 \left[\hat{S}_{12}(2) \ket{0 0}\bra{0 0} \hat{S}_{12}^\dag(2)\right]
            \hat{D}_{1}^\dag\left(\sqrt{2}\alpha\right).
\end{split}
\end{align}
However, by definition, the partial trace above is nothing but $\mathcal{A}_2(\ket{0}\bra{0})$, which is known \cite{QLA_De_Palma_2017} to be just the thermal state 
\begin{align} \label{eq:A_displacement_semi-commutation}
\begin{split}
    \hat{g}_\beta \coloneqq 
        \frac{e^{-\beta \hat{a}^\dag\hat{a}}}{\Tr e^{-\beta \hat{a}^\dag\hat{a}}}
\end{split}
\end{align}
with inverse temperature $\beta=\ln 2$. This yields
\begin{align} \label{eq:A_commutation_with_displacement}
\begin{split}
    \mathcal{A}_2(\ket{\alpha}\bra{\alpha})
        = \hat{D}\left(\sqrt{2}\alpha\right)\hat{g}_{\ln 2}\hat{D}^\dag\left(\sqrt{2}\alpha\right).
\end{split}
\end{align}

We now move to the phase-space. The P representation of thermal states is \cite{optics_Agarwal_2012}
\begin{align}
\begin{split}
    P_{\hat{g}_{\beta}}(\gamma) = \frac{1}{\bar{n}}e^{-|\gamma|^2/\bar{n}},
\end{split}
\end{align}
where $\bar{n}=1/(e^{\beta}-1)$ is the mean number of photons in the state. For our case, i.e. $\beta=\ln 2$, it is easy to see that $\bar{n}=1$, resulting in
\begin{align}
\begin{split}
    \mathcal{A}_2(\ket{\alpha}\bra{\alpha})
        &= \int \frac{d^2\gamma}{\pi}e^{-|\gamma|^2}
        \hat{D}\left(\sqrt{2}\alpha\right)\ket{\gamma}\bra{\gamma}
        \hat{D}^\dag\left(\sqrt{2}\alpha\right)\\
        &= \int \frac{d^2\gamma}{\pi}e^{-|\gamma|^2}
        \ket{\gamma+\sqrt{2}\alpha}\bra{\gamma+\sqrt{2}\alpha},
\end{split}
\end{align}
where we used the defining property of the displacement operator: displacing coherent states. Changing the integration variable, the above reduces to
\begin{align} \label{eq:A_on_coherent_states}
\begin{split}
    \mathcal{A}_2(\ket{\alpha}\bra{\alpha})
        = \int \frac{d^2\gamma}{\pi}e^{-|\gamma-\sqrt{2}\alpha|^2}
        \ket{\gamma}\bra{\gamma}.
\end{split}
\end{align}
To obtain $\mathcal{C}(\ket{\alpha}\bra{\alpha})$, we need only to apply $\mathcal{E}_{1/2}$ to the above, i.e.
\begin{align}
\begin{split}
    \mathcal{C}(\ket{\alpha}\bra{\alpha})
        = \int \frac{d^2\gamma}{\pi}e^{-|\gamma-\sqrt{2}\alpha|^2}
        \mathcal{E}_{1/2}(\ket{\gamma}\bra{\gamma}).
\end{split}
\end{align}
However, the action of the quantum limited attenuator on coherent states is simple \cite{QLA_De_Palma_2017}:
\begin{align} \label{eq:E_on_coherent_states}
    \mathcal{E}_\lambda \left(\ket{\alpha}\bra{\alpha}\right) 
        = \ket{\sqrt{\lambda}\alpha}\bra{\sqrt{\lambda}\alpha}.
\end{align}
Therefore, we have
\begin{align}
\begin{split}
    \mathcal{C}(\ket{\alpha}\bra{\alpha})
        = \int \frac{d^2\gamma}{\pi}e^{-|\gamma-\sqrt{2}\alpha|^2}
        \ket{\gamma/\sqrt{2}}\bra{\gamma/\sqrt{2}}
\end{split}
\end{align}
and, ultimately, after a change of integration variable,
\begin{align} \label{eq:proof_C_coherent_final}
\begin{split}
    \mathcal{C}(\ket{\alpha}\bra{\alpha})
        = 2 \int \frac{d^2\gamma}{\pi}e^{-2|\gamma-\alpha|^2}
        \ket{\gamma}\bra{\gamma}.
\end{split}
\end{align}
Note that $2e^{-2|\gamma-\alpha|^2}$ is actually the Wigner function of a coherent state, meaning that with the above equation, we have essentially proved Eq. (\ref{eq:P_into_W_C}) for coherent state inputs.

To generalize to arbitrary operators, we once again employ the P representation, obtaining
\begin{align}
\begin{split}
    \mathcal{C}\big(\hat{X}\big)
        & = \int \frac{d^2\alpha}{\pi} P_{\hat{X}}(\alpha)\,
            \mathcal{C}(\ket{\alpha}\bra{\alpha}).          
\end{split}
\end{align}
However, we already know $\mathcal{C}(\ket{\alpha}\bra{\alpha})$ from Eq. (\ref{eq:proof_C_coherent_final}). We thus get
\begin{align}
\begin{split}
    \mathcal{C}\big(\hat{X}\big)
        & = \int \frac{d^2\gamma}{\pi} \ket{\gamma}\bra{\gamma} \times 2
            \int \frac{d^2\alpha}{\pi} P_{\hat{X}}(\alpha) e^{-2|\gamma-\alpha|^2}.       
\end{split}
\end{align}
Comparing the rightmost integral in the above equation with Eq. (\ref{eq:P_into_W_transform}), we can see that
\begin{align}
\begin{split}
    \mathcal{C}\big(\hat{X}\big)
        & = \int \frac{d^2\gamma}{\pi} W_{\hat{X}}(\gamma) \ket{\gamma}\bra{\gamma}.
\end{split}
\end{align}
Clearly, then, by definition (\ref{eq:P}), the P distribution of $\mathcal{C}\big(\hat{X}\big)$ is the W distribution of $\hat{X}$. This is exactly what we wanted to prove.
\end{proof}

\section{Discussion}
\label{sec:discussion}
Let us discuss our results. To start with, our equations (\ref{eq:P_into_W_C}-\ref{eq:P_into_Q_C}) provide a direct operational interpretation for the mathematical relations (\ref{eq:P_into_W_transform}-\ref{eq:P_into_Q_transform}). As we showed, the Gaussian smoothing appearing there corresponds physically to a combination of quantum limited amplification and attenuation. Notably, both channels enter the relations with finite strengths, meaning that, as Gaussian operations, they can be relatively easily performed in a modern quantum optical laboratory \cite{Gaussianity_resource_Albarelli_2018,Gaussianity_resource_Takagi_2018}. This improves upon a previously known relation between the P and Q distributions \cite{quantum_phase_Q_amplification_Schleich_1992}, which requires infinite amplification and thus infinite energy, and is therefore unphysical.

Our findings also strengthen the intuition that the Q distribution can be considered more classical than the W distribution, which in turn is more classical than the P distribution. As mentioned previously, the quantum limited amplifier is already known for making various phenomena more classical: besides the previously mentioned transformation of the P distribution into the Q distribution \cite{quantum_phase_Q_amplification_Schleich_1992}, it also transforms the von Neumann entropy into the semi-classical Wehrl entropy \cite{Wehrl_entropy_De_Palma_2017} and the Pegg-Barnett quantum phase formalism into the Husimi distribution-based Paul formalism \cite{Pegg-Barnett_Paul_relation_Linowski}. Furthermore, both the quantum limited amplifier and the quantum limited attenuator are irreversible channels, meaning that the application of the channel $\mathcal{C}$ to the system results in irreversible coarse-graining, i.e. loss of the full quantum information. 

Notably, unlike for the quantum limited amplifier $\mathcal{A}_2$ or the quantum limited attenuator $\mathcal{E}_{1/2}$, where the coarse-graining comes at the price of a radical change in the number of particles in the system (roughly either doubling or halving it), their combined action through the channel $\mathcal{C}$ leaves this number effectively intact for all but very low particle numbers. By combining the formulas (\ref{eq:n_A}, \ref{eq:n_E}) one can easily calculate that under the action of the channel, the mean particle number of the system transforms as
\begin{align} \label{eq:n_C}
    \braket{\hat{a}^\dag\hat{a}}_{\mathcal{C}(\hat{\rho})} = \braket{\hat{a}^\dag\hat{a}}_{\hat{\rho}} + 1.
\end{align}
The additional term of $1$ can be considered negligible for all systems but those very close to the vacuum state.

The classicality-increasing effect of the channel $\mathcal{C}$ can be seen by considering sets of semi-classical states. Let us define by $\mathcal{S}$ the set of all quantum states associated with the considered Hilbert space and by $\mathcal{S}_{W_+}$, $\mathcal{S}_{P_+}$ the subsets of $\mathcal{S}$ consisting of all the states with non-negative Wigner and Glauber-Sudarshan distributions, respectively. We have that
\begin{align}
    \mathcal{S} \supset \mathcal{S}_{W_+} \supset \mathcal{S}_{P_+},
\end{align}
with the classicality of the subsequent sets increasing from left to right (since the non-negativity of the P distribution is a stronger criterion for classicality than its analog for the W function). Now, consider the image $\Im(\mathcal{C})$ of the channel $\mathcal{C}$ when acting on $\hat{\rho}\in\mathcal{S}$. Because the Q distribution is non-negative for all density operators, it follows directly from our main result that $\Im(\mathcal{C})\subset\mathcal{S}_{W_+}$ and $\Im(\mathcal{C}^2)\subset\mathcal{S}_{P_+}$: subsequent actions of the channel $\mathcal{C}$ take the quantum state into more and more semi-classical sets.

Let us observe that, because the action of the channel $\mathcal{C}^2$ is equivalent to replacing the P distribution of the input by its Q distribution, the channel must necessarily coincide with the projection onto the set of coherent states:
\begin{align} \label{eq:C_squared_projection}
    \mathcal{C}^2\big(\hat{X}\big) = \int \frac{d^2\alpha}{\pi} 
        \ket{\alpha}\bra{\alpha} \hat{X} \ket{\alpha}\bra{\alpha}.
\end{align}
Using the P representation with Eqs. (\ref{eq:E_on_coherent_states}, \ref{eq:A_on_coherent_states}, \ref{eq:P_into_Q_transform}), one can also easily derive yet another representation for the channel (for convenience, we present the full derivation in Appendix \ref{app:C_squared_decomposition}):
\begin{align} \label{eq:C_squared_decomposition}
    \mathcal{C}^2 = \mathcal{A}_{2}\circ\mathcal{E}_{1/2}.
\end{align}
In other words, $\mathcal{C}^2$ has identical components as the channel $\mathcal{C}$, but applied in reverse order. These findings have two worthwhile consequences:
\begin{itemize}
    \item The projection onto the set of coherent states, effectively a projection onto a set of semi-classical pure states, can be experimentally implemented either by a double application of the channel $\mathcal{C}=\mathcal{E}_{1/2}\circ\mathcal{A}_{2}$ or by a single application of the channel $\mathcal{C}^2=\mathcal{A}_{2}\circ\mathcal{E}_{1/2}$;
    \item The channel $\mathcal{C}$ can be regarded as a square root of such a projection, i.e. a square root of a projection onto the set of semi-classical pure states.
\end{itemize}
We remark that channels of the form 
\begin{align}
    \mathcal{N}_E = \mathcal{A}_{E+1}\circ\mathcal{E}_{1/(E+1)},
\end{align}
which realize $\mathcal{C}^2$ for $E=1$, are known in the literature under the name of additive-noise channels \cite{pure-loss_De_Palma_2017}.

To demonstrate that the action of the channel $\mathcal{C}$ indeed makes its input more classical, let us consider the displaced parity operator \cite{displaced_parity_Birrittella_2021,displaced_parity_Rundle_2021,displaced_parity_Bishop_1994}:
\begin{align} \label{eq:parity_operator}
    \hat{\Pi}(\alpha)\coloneqq 2 \hat{D}({\alpha})(-1)^{\hat{a}^\dag\hat{a}}\hat{D}^\dag(\alpha),
\end{align}
used, e.g. in spectroscopy. Here, the multiplicative factor of 2 was added to normalize the operator to one \footnote{A careful reader may find this statement rather dangerous: after all, due to the cyclic property of trace, $\Tr\hat{\Pi}(\alpha)=2\Tr(-1)^{\hat{a}^\dag\hat{a}}$. If we compute this in the Fock basis, we obtain $\Tr(-1)^{\hat{a}^\dag\hat{a}} = 1-1+1-1+\ldots$, which is indefinite. However, this is simply a sign that we chose a wrong approach to the problem. For example, if we use the coherent state basis, we quickly find $\Tr(-1)^{\hat{a}^\dag\hat{a}}=1/2$. The same conclusion can be reached by other means, see, e.g. \cite{displaced_parity_Bishop_1994}.}. As we explicitly calculate in Appendix \ref{app:parity_C}, despite not even being a valid quantum state (due to its negative eigenvalues), the displaced parity operator becomes semi-classical after subjecting it once and twice to the channel~$\mathcal{C}$. After a single application, we obtain the coherent state:
\begin{align} \label{eq:parity_C}
    \mathcal{C}\big[\hat{\Pi}(\alpha)\big] &= \ket{\alpha}\bra{\alpha},
\end{align}
which is of course regarded as one of the most classical states in quantum mechanics. However, it can be argued that the state is still partially quantum due to its ideal purity -- its von Neumann entropy $S_V(\hat{\rho})\coloneqq -\Tr(\hat{\rho}\ln\hat{\rho})$ vanishes, indicating zero uncertainty about the quantum system it describes. This stops being the case after a second application of the channel:
\begin{align} \label{eq:parity_C_squared}
    \mathcal{C}^2\big[\hat{\Pi}(\alpha)\big] &= 2\int \frac{d^2\gamma}{\pi}
        e^{-2|\alpha-\gamma|^2}\ket{\gamma}\bra{\gamma},
\end{align}
where we obtain a classical Gaussian mixture of (already semi-classical) coherent states.

The fact that coherent states can be written as the action of $\mathcal{C}$ on the parity operator has some further implications. Namely, using the P representation, we have that, for any operator
\begin{align}
\begin{split}
    \hat{X} = \mathcal{C}
        \left[\int \frac{d^2\alpha}{\pi} P_{\hat{X}}(\alpha)\hat{\Pi}(\alpha)\right].
\end{split}
\end{align}
But if this is true, then we can identify the following formal expression for the channel inverse:
\begin{align} \label{eq:C_inverse}
\begin{split}
    \mathcal{C}^{-1}\big(\hat{X}\big) = 
        \int \frac{d^2\alpha}{\pi} P_{\hat{X}}(\alpha)\hat{\Pi}(\alpha).
\end{split}
\end{align}
We stress that this formula should be treated with care: although the operator above indeed behaves like the channel inverse (e.g. applying $\mathcal{C}$ to the above gives $\hat{X}$), it may not be unique. Furthermore, and perhaps more importantly, the above formal inverse may in practice not exist for certain inputs, in the same way that, although formally the P distribution may be obtained from the Q distribution by inverting the Weierstrass transform (\ref{eq:P_into_Q_transform}), in practice, the result is often divergent.

The example of the parity operator highlights another subtlety of the ``classicalization'' procedure through the channel $\mathcal{C}$. According to Proposition \ref{th:main}, if we apply the channel $\mathcal{C}$ (or $\mathcal{C}^2$) to any valid density operator, regardless of how much ``quantumness'' it contains, we will end up with another valid density operator: a completely semi-classical one in the sense of having a non-negative W (or P) distribution. However, the converse is in general not true: if we start with a semi-classical state, it will in general not be possible to write it as the action of $\mathcal{C}$ (or $\mathcal{C}^2$) on some other valid density operator. An explicit example of this is given by the state (\ref{eq:parity_C_squared}): although it is semi-classical, it originates through the channel $\mathcal{C}^2$ by the displaced parity operator (\ref{eq:parity_operator}), which does not describe a quantum state.

This property allows us to define new formal criteria for non-classicality based on negativity of the Wigner and Glauber-Sudarshan distributions. Using the channel inverse (\ref{eq:C_inverse}) in Eq. (\ref{eq:W_into_Q_C}), we immediately find that for any state~$\hat{\rho}$
\begin{align}
    W_{\hat{\rho}}(\alpha) = Q_{\mathcal{C}^{-1}(\hat{\rho})}(\alpha).
\end{align}
Since $\mathcal{C}^{-1}(\hat{\rho})$ does not have to be a valid density operator, its Husimi distribution may have negative values, corresponding to a non-classical Wigner function of $\hat{\rho}$. This implies the following conditions:
\begin{enumerate}
    \item A sufficient condition for classicality of the state $\hat{\rho}$ with respect to Wigner function negativity is that $\mathcal{C}^{-1}(\hat{\rho})$ is positive semidefinite, i.e. it is a valid quantum state.
    \item A necessary condition for non-classicality of the state $\hat{\rho}$ with respect to Wigner function negativity is that $\mathcal{C}^{-1}(\hat{\rho})$ is not positive semidefinite.    
\end{enumerate}
Note that even if $\mathcal{C}^{-1}(\hat{\rho})$ is not positive semidefinite, it can still have a non-negative Husimi distribution, which is why the second criterion given above is only necessary for non-classicality.

Similar conditions can be constructed for the Glauber-Sudarshan distribution. Due to Eq. (\ref{eq:P_into_Q_C}), we have
\begin{align}
    P_{\hat{\rho}}(\alpha) = Q_{\mathcal{C}^{-2}(\hat{\rho})}(\alpha),
\end{align}
meaning that:
\begin{enumerate}
    \item A sufficient condition for classicality of the state $\hat{\rho}$ with respect to Glauber-Sudarshan distribution negativity is that $\mathcal{C}^{-2}(\hat{\rho})$ is positive semidefinite.
    \item A necessary condition for non-classicality of the state $\hat{\rho}$ with respect to Glauber-Sudarshan distribution negativity is that $\mathcal{C}^{-2}(\hat{\rho})$ is not positive semidefinite.    
\end{enumerate}

The above criteria could even give rise to non-classicality measures based on the degree to which $\mathcal{C}^{-1}(\rho)$ and $\mathcal{C}^{-2}(\rho)$ fail to be positive semidefinite, similarly to how it is sometimes possible to measure state entanglement based on the degree to which its partial transpose is negative \cite{PPT,PPT_cv_systems,two-mode_gaussian_etc_proper_norm}. Alternatively, one could consider the distance (such as the trace distance or relative entropy) of state $\hat{\rho}$ from the set of states of the form $\mathcal{C}(\hat{\sigma})$ [or $\mathcal{C}^2(\hat{\sigma})$], with $\hat{\sigma}$ being valid quantum states. 
As open problems for future research we leave both: in-depth formal analysis of the above idea within the framework of resource theories, as well as its practical applications for selected families of states (e.g. with high symmetry).

\section{Outlooks}
\label{sec:outlooks}
We derived explicit operational relations between the three most well-known quantum quasiprobability distributions: the Glauber-Sudarshan distribution and the Wigner and Husimi functions. Notably, these relations, summarized in Proposition \ref{th:main}, rely fully on the channel $\mathcal{C}$, a single composition of finite-strength quantum limited amplifier and quantum limited attenuator, both readily available in a modern quantum optical laboratory. Our results shed light on the operational understanding of the quantum-to-classical transition.

Our findings may provide a basis for further research. Most notably, as discussed in the last section, the channel $\mathcal{C}$ can be used to define new formal criteria for and measures of non-classicality of quantum states. It would be interesting to see whether these criteria and measures could be employed in practical calculations or if they can be related to other known objects of this type. Furthermore, one could verify experimentally what happens to quantum systems when subjected to the channels $\mathcal{C}$ and $\mathcal{C}^2$, with special emphasis on systems with strong quantum features, such as a high degree of entanglement or purity.

\begin{acknowledgements}
We acknowledge support by the Foundation for Polish Science (International Research Agenda Programme project, International Centre for Theory of Quantum Technologies, Grant No. 2018/MAB/5, cofinanced by the European Union within the Smart Growth Operational Program).
\end{acknowledgements}

\bibliography{report}
\bibliographystyle{obib}

\appendix
\section{Proof of Eq. (\ref{eq:C_squared_decomposition})} 
\label{app:C_squared_decomposition}
\setcounter{equation}{0}
\renewcommand{\theequation}{\ref{app:C_squared_decomposition}\arabic{equation}}
In this appendix, we prove Eq. (\ref{eq:C_squared_decomposition}). Using the P representation and Eq. (\ref{eq:E_on_coherent_states}), we immediately get
\begin{align}
    \mathcal{A}_{2}\circ\mathcal{E}_{1/2}\big(\hat{X}\big) =
        \int \frac{d^2\alpha}{\pi} P_{\hat{X}}(\alpha)
        \mathcal{A}_{2}\left(\ket{\alpha/\sqrt{2}}\bra{\alpha/\sqrt{2}}\right).
\end{align}
Employing (\ref{eq:A_on_coherent_states}) and rearranging, this becomes
\begin{align}
    \mathcal{A}_{2}\circ\mathcal{E}_{1/2}\big(\hat{X}\big) =
        \int \frac{d^2\gamma}{\pi}\ket{\gamma}\bra{\gamma}
        \times \int \frac{d^2\alpha}{\pi} P_{\hat{X}}(\alpha) e^{-|\gamma-\alpha|^2}.
\end{align}
Comparing with Eq. (\ref{eq:P_into_Q_transform}), we immediately see that the second integral is just the Q distribution of $\hat{X}$. Since the latter is arbitrary, the channel $\mathcal{A}_{2}\circ\mathcal{E}_{1/2}$ must coincide with the r.h.s. of Eq. (\ref{eq:C_squared_projection}), concluding the proof.

\section{Proof of Eqs. (\ref{eq:parity_C}, \ref{eq:parity_C_squared})} 
\label{app:parity_C}
\setcounter{equation}{0}
\renewcommand{\theequation}{\ref{app:parity_C}\arabic{equation}}
In this appendix, we calculate the action of the channels $\mathcal{C}$, $\mathcal{C}^2$ on the displaced parity operator (\ref{eq:parity_operator}). First, let rewrite Eq. (\ref{eq:A_commutation_with_displacement}) as [we remind the reader that $\hat{g}_{\ln 2}=\mathcal{A}_2\left(\ket{0}\bra{0}\right)$]
\begin{align} \label{eq:A_commutation_with_displacement_rewritten}
\begin{split}
    \mathcal{A}_2\left[\hat{D}(\alpha)\ket{0}\bra{0}\hat{D}^\dag(\alpha)\right]
    = \hat{D}\left(\sqrt{2}\alpha\right)
        \mathcal{A}_2\left(\ket{0}\bra{0}\right)
        \hat{D}^\dag\left(\sqrt{2}\alpha\right).
\end{split}
\end{align}
As we can see, at least for the initial vacuum state, the displacement channel and the quantum limited amplifier nearly commute, i.e. they commute if we rescale the displacement factor by $\sqrt{2}$. However, it is not difficult to convince oneself that our derivation of Eq. (\ref{eq:A_commutation_with_displacement}), and hence (\ref{eq:A_commutation_with_displacement_rewritten}), made no use of the fact that the initial state was the vacuum state. This means that for arbitrary $\hat{X}$, the following holds:
\begin{align} 
\begin{split}
    \mathcal{A}_2\left[\hat{D}(\alpha)\hat{X}\hat{D}^\dag(\alpha)\right]
    = \hat{D}\left(\sqrt{2}\alpha\right)
        \mathcal{A}_2\big(\hat{X}\big)
        \hat{D}^\dag\left(\sqrt{2}\alpha\right).
\end{split}
\end{align}
In particular, for the displaced parity operator,
\begin{align}
\begin{split}
    \mathcal{A}_2\big[\hat{\Pi}(\alpha)\big]
    = 2\hat{D}\left(\sqrt{2}\alpha\right)
        \mathcal{A}_2\left[(-1)^{\hat{a}^\dag\hat{a}}\right]
        \hat{D}^\dag\left(\sqrt{2}\alpha\right).
\end{split}
\end{align}
The action of the quantum limited amplifier can be calculated from the explicit formula in the number basis \cite{Pegg-Barnett_Paul_relation_Linowski}
\begin{align} \label{eq:QLA_standard_basis}
\begin{split}
    \mathcal{A}_{\kappa}\big(\hat{X}\big)=&\:\frac{1}{\kappa}\sum_{\substack{j=0}}^\infty
        \left(\frac{\kappa-1}{\kappa}\right)^{j}
        \sum_{\substack{m,n=0}}^\infty
        \frac{X_{mn}}{\sqrt{\kappa}^{m+n}}\\
    &\sqrt{\binom{j+m}{j}\binom{j+n}{j}}\ket{j+m}\bra{j+n},
\end{split}
\end{align}
which yields
\begin{align}
\begin{split}
    \mathcal{A}_2\left[(-1)^{\hat{a}^\dag\hat{a}}\right] = \frac{1}{2}\ket{0}\bra{0}
\end{split}
\end{align}
and hence
\begin{align}
\begin{split}
    \mathcal{A}_2\big[\hat{\Pi}(\alpha)\big]
        = \hat{D}\left(\sqrt{2}\alpha\right)\ket{0}\bra{0}\hat{D}^\dag\left(\sqrt{2}\alpha\right) 
        = \ket{\sqrt{2}\alpha}\bra{\sqrt{2}\alpha}.
\end{split}
\end{align}
Applying $\mathcal{E}_{1/2}$ to both sides and using (\ref{eq:E_on_coherent_states}), we obtain Eq. (\ref{eq:parity_C}).

It remains to prove Eq. (\ref{eq:parity_C_squared}). However, according to what we have just shown, $\mathcal{C}^2\big[\hat{\Pi}(\alpha)\big]=\mathcal{C}(\ket{\alpha}\bra{\alpha})$, which we have already calculated in Eq. (\ref{eq:proof_C_coherent_final}). As we can easily see, the latter indeed coincides with Eq. (\ref{eq:parity_C_squared}).

\end{document}